\documentclass[11pt,a4paper]{article}

\usepackage{times}
\usepackage{soul}
\usepackage{url}
\usepackage[hidelinks]{hyperref}
\usepackage[utf8]{inputenc}
\usepackage[small]{caption}
\usepackage{graphicx}
\usepackage{amsmath,amssymb,amsfonts,amsthm}
\usepackage{mathtools}
\usepackage{amsthm}
\usepackage{booktabs}
\usepackage{algorithm}
\usepackage{xspace}
\usepackage{authblk}

\usepackage[switch]{lineno}
\usepackage[noend]{algpseudocode}

\usepackage{cancel}

\usepackage{tikz}
\usepackage{todonotes}
\usetikzlibrary{shapes,shapes.callouts,shadows,arrows,backgrounds,%
matrix,patterns,arrows,decorations.pathmorphing,decorations.pathreplacing,%
positioning,fit,calc,decorations.text%
}

\theoremstyle{plain}
    \newtheorem{theorem}{Theorem}
    
    \newtheorem{lemma}[theorem]{Lemma}

    \newtheorem*{theorem*}{Theorem}
    \newtheorem*{corollary*}{Corollary}
    \newtheorem*{lemma*}{Lemma}
    \newtheorem*{proposition*}{Proposition}
    \newtheorem*{claim*}{Claim}
    \newtheorem*{conjecture*}{Conjecture}

\theoremstyle{definition}

    \newtheorem*{definition*}{Definition}
    \newtheorem*{example*}{Example}
        \newtheorem*{obs*}{Observation}

\theoremstyle{remark}

    \newtheorem*{remark*}{Remark}
    \newtheorem*{note*}{Note}

\DeclarePairedDelimiter\abs{\lvert}{\rvert}

\mathchardef\mhyphen="2D

\def\A{{\bf A}}
\def\B{{\bf B}}

\def\S{{\bf S}}

\DeclareSymbolFont{bbold}{U}{bbold}{m}{n}
\DeclareSymbolFontAlphabet{\mathbbold}{bbold}

\newcommand*{\integers}{\mathbb{Z}}
\newcommand*{\rationals}{\mathbb{Q}}
\newcommand*{\naturals}{\mathbb{N}}

\newcommand*{\csp}{\textsc{CSP}}
\newcommand*{\maxcsp}{\textsc{MaxCSP}}
\newcommand*{\probstp}{\textsc{STP}}
\newcommand*{\maxstp}{\textsc{MaxSTP}}
\newcommand*{\maxsat}{\textsc{MaxSAT}}

\newcommand*{\variables}{V}
\newcommand*{\constraints}{C}

\newcommand*{\stp}{\stpfull^{\leqslant}\xspace}

\algrenewcommand{\algorithmiccomment}[1]{\hfill$//$ \textit{#1}}

\newcommand{\cc}[1]{{\mbox{\textnormal{\textsf{#1}}}}\xspace}  

\newcommand{\Nat}{\mathbb{N}}
\newcommand{\bigoh}{O}
\newcommand{\PP}{\cc{P}}
\newcommand{\NP}{\cc{NP}}

\newcommand{\FPT}{\cc{FPT}}
\newcommand{\XP}{\cc{XP}}
\newcommand{\Weft}{{\cc{W}}}
\newcommand{\W}[1]{{\Weft}{\cc{[#1]}}}
\newcommand{\paraNP}{\cc{pNP}}

\newcommand{\hy}{\hbox{-}\nobreak\hskip0pt}

\def\stp{{\bf S}\xspace}

\newcommand*{\magni}[1]{{\sf mag}(#1)\xspace}

\newcommand*{\tw}{{\sf tw}\xspace}
\newcommand*{\vc}{{\sf vc}\xspace}

\usepackage{boxedminipage}
\newcommand{\probfont}[1]{{\sc #1}\xspace}

\newcommand{\MCC}{\probfont{Multicolor Clique}}

\newcommand{\pbDef}[3]{%
\noindent
\begin{center}
\begin{boxedminipage}{0.98 \columnwidth}
#1\\[5pt]
\begin{tabular}{l p{0.70 \columnwidth}}
Input: & #2\\
Question: & #3
\end{tabular}
\end{boxedminipage}
\end{center}
}

\newcommand{\pbDefP}[4]{%
\noindent
\begin{center}
\begin{boxedminipage}{0.98 \columnwidth}
#1\\[5pt]
\begin{tabular}{l p{0.70 \columnwidth}}
Input: & #2\\
Param.: & #3\\
Question: & #4
\end{tabular}
\end{boxedminipage}
\end{center}
}

\mathchardef\mhyphen="2D

\def\integers{\mathbb{Z}}
\def\rationals{\mathbb{Q}}
\def\naturals{\mathbb{N}}

\def \pi  {{\sf pi}\xspace}

\title{Maximum Satisfiability of\\ Simple Temporal Problems}

\author{Johannes K. Fichte}
\author{Johanna Groven}
\author{Peter Jonsson}
\author{Victor Lagerkvist}
\author{Jorke M. de Vlas}
\affil{Department of Computer and Information Science (IDA)\\ Link\"oping University \\ $\;$ \\ {\tt firstname.lastname@liu.se}}

\date{}

\newcommand{\TTT}{\ensuremath{\mathcal{T}}}
\usepackage{microtype}
\begin{document}

\maketitle

\begin{abstract}

The Simple Temporal Problem (\probstp) is a core framework for quantitative temporal constraints. As \probstp\ data can be inconsistent, we study $\maxstp$: compute a maximum-cardinality consistent subset of constraints. This extension is \NP-hard, and we analyze its parameterized complexity under measures that capture practically relevant instance features: the number of variables $n$ (instance scale), the maximum coefficient magnitude $k$ (numeric range), and structural parameters of the constraint graph such as treewidth $\tw$ (decomposability) and vertex cover size $\vc$ (density).
We show that $\maxstp$ is \W{1}-hard parameterized by $n$, implying that $n$ and parameters that depend on $n$ (including $\tw$ and $\vc$) are insufficient for fixed-parameter tractability. For combined parameters, we give an $O^*(k^n)$-time algorithm, yielding single-exponential solvability for fixed $k$. While $k+\tw$ remains \W{1}-hard, $\maxstp$ is in \XP via an $O^*((n\cdot k)^{\tw})$ algorithm.
Our results suggest that \maxstp\ is often computationally harder than optimizing qualitative CSPs---we verify that many such problems (including RCC-8 and Allen's algebra) are \FPT\ when parameterized by $n$ or $\tw$. However, we also demonstrate that \FPT\ algorithms for \maxstp\ are indeed possible but with other parameters such as $k + \vc$. 
\end{abstract}

\section{Introduction}
The \emph{simple temporal problem}~\cite{Dechter:etal:ai91} (\probstp) is a fundamental quantitative temporal constraint problem in AI. The task is to decide if a set of constraints (for $a, b \in \mathbb{Z} \cup \{-\infty,\infty\}$) of the form $a \leq x - y \leq b$, or $a \leq x \leq b$, is consistent, in the sense that it admits at least one satisfying assignment. One of the original motivations for introducing this formalism is that it is a simple and efficient way that makes it possible to detect inconsistencies in temporal data. The \probstp\ problem is the most well-known temporal framework with thousands of reported applications~\cite{hunsberger_et_al:LIPIcs.TIME.2021.1}.
Early success stories include e.g.\ autonomous vehicles in space exploration~\cite{722362,fukunaga}, and temporal planning~\cite[Chapter 14]{planningbook} while more recent examples include e.g. managing temporal data in robotic scheduling~\cite{wang2022} and multi-agent systems~\cite{10.5555/2566972.2566976}. 

\probstp\ can be decided in polynomial time but has limited expressive power for handling uncertainty. To cope with this several extensions have been considered, e.g., by allowing disjunctions~\cite{Dechter:etal:ai91,Tsamardinos:Pollack:ai2003}, and, more recently~\cite{dabrowski2022}, by allowing a small set of inconsistencies. 
 Thus, we want to identify a maximum consistent subset of the data: this is
the \emph{maximum constraint satisfaction problem} (\probfont{MaxCSP}). Such problems are well-studied in the AI literature
  (see, for instance, see~\cite{condotta2016local,condotta2016sat,chomicki2005minimal,bertossi2004query} and the references therein).
Its Boolean analogue is the ubiquitous \maxsat\ problem where 
clauses in Boolean CNF are considered~\cite{BacchusJarvisaloMartins21}.
An annual solver evaluation (\maxsat\ Evaluation) collects applications, and instances, and tests solvers practically~\cite{BergEtAl24}.

Hence, we study the problem $\maxstp$: given an \probstp\ instance and $m \geq 0$, decide if it possible to obtain a satisfiable instance by keeping at least $m$ constraints. 
We first observe that \maxstp, in contrast to STP, is \NP-hard~\cite{dabrowski2022} so we shift to a \emph{parameterized} complexity view where we identify parameters that correlate to a ``hidden structure'' so that the problem can be solved efficiently. In the best case, we can solve a problem in $f(k) \cdot ||I||^{O(1)}$ time, where $I$ is an instance with parameter $k \in \mathbb{N}$, $||I||$ the number of bits required to represent it, and $f \colon \mathbb{N} \to \mathbb{N}$ a computable function. A problem admitting such an algorithm is said to be \emph{fixed parameter tractable} (\FPT). The parameterized complexity of \maxsat\ is well-studied (cf.~\cite{Dell:etal:algorithmica2017}) and recently
also solvers that employ and improve on parameterized techniques (such as \emph{treewidth}) have been introduced~\cite{BannachHecher24}. However, there is little reason to believe that these results carry over to the STP setting. In fact, we will establish that $\maxstp$ is fundamentally different.

For \maxstp\ we may immediately remark that the parameter $m$ in itself hardly makes sense since we would expect it to grow with the instance size. Here, we diverge from Dabrowski et al.~\cite{dabrowski2022} who considered the dual problem \textsc{AlmostCSP} where the task is instead to delete at most $m$ constraints --- the problems are clearly classically equivalent but differ under parameterized and approximation complexity. Instead, we propose the \emph{magnitude} $\magni{C}$, the largest absolute value of any constant appearing in an instance $(V,C)$. Towards \FPT\ algorithms we first (in Section~\ref{sec:unbounded_magnitude}) investigate the parameterized complexity of $\maxstp$ with parameter $n = |V|$. Observe that this for many problems, e.g., finite-domain CSPs, and qualitative reasoning problems such as \emph{Allen's interval algebra} and the \emph{region connection calculus}, is an incredibly strong parameter that gives trivial \FPT\ algorithms by exhaustive enumeration. This is not possible for \maxstp\, but we manage to construct an algorithm with a running time dominated by $\magni{C}^n$, i.e., for fixed magnitude $\magni{C}$ the problem is solvable in single-exponential time.
We match this with a sharp  \W{1}-hardness proof (under parameter $n$) via a \emph{Sidon set}  construction from the \textsc{multicolor clique} problem. Under the conjecture that 3-SAT is not solvable in subexponential time (the \emph{exponential-time hypothesis} (ETH) we can even infer that $\maxstp$ is not solvable in $2^{o(n\log \magni{C} + f(n))}$ time for any computable function $f \colon \naturals \to \naturals$, i.e., we have matching upper and lower bounds for $\maxstp$ when analyzed with $n$. 
An important consequence of this lower bound is that it rules out \FPT\ algorithms for e.g.\ graph parameters that depend on $n$, such as treewidth, showing the need to consider multiple parameters to reach \FPT. We analyze this in Section~\ref{sec:boundedmagnitude} and consider magnitude combined with either treewidth (\tw), or the size of a minimal vertex cover (\vc). Here, we begin with the weaker parameter \vc and construct a branching based \FPT\ algorithm that exploits the parameters to bound the domain of each variable. Treewidth turns out to be less favourable and, building on the construction in Section~\ref{sec:unbounded_magnitude},  we give a \W{1}-hardness proof. Treewidth still turns out to be useful, however, and while insufficient for \FPT, it can still be used to obtain an \XP\ algorithm. 

Put together, $\maxstp$ is fairly resilient, but not immune, to \FPT\ algorithms, and parameters that normally work (e.g., $n = |V|$ or $\tw$) lead to \W{1}-hardness.  In Section~\ref{sec:qualitative} we compare this to the corresponding question for \emph{qualitative} reasoning, i.e., frameworks where we only care about the relations between variables, and not their precise numerical values. \emph{Allen's interval algebra}, the \emph{region connection calculus}, and the \emph{cardinal direction calculus}, are all well-known examples of \NP-hard qualitative reasoning problems. They can be formulated as infinite-domain CSPs for ``well-behaved'' templates that often makes it possible to use parameters and techniques from the finite-domain case. Here, the gap between $\csp$ and $\maxcsp$ is much more narrow and DP-style algorithms for the former can often be extended to the latter. We exemplify a general criterion for when this is possible which, for example, captures Allen's interval algebra $\csp(\mathrm{Allen})$, where we get \FPT\ for $\maxcsp(\mathrm{Allen})$ with parameter $\tw$ as well as an improved DP algorithm under parameter $n$. Hence, the $\maxcsp$ problem seems to be fundamentally different for qualitative versus quantitative reasoning.

\section{Preliminaries}

We begin by presenting the notation, the (maximum) CSPs, and the most relevant concepts from parameterized complexity. To express running times, we sometimes use the notation $O^*(\cdot)$ that hides factors polynomial in the input size.

\subsection{Constraint Satisfaction}
\label{sec:csp}

A \emph{constraint language} $\A$ is a set of relations over a set $D$ (a \emph{domain}).
  The {\em constraint satisfaction problem} over $\A$ ($\csp(\A)$)
  is defined as follows:

\pbDef{\csp(\A)}
{A tuple $(\variables,\constraints)$, where $\variables$ is a set
      of variables and $\constraints$ is a multiset of constraints of the form
      $R(v_1, \dots, v_a)$, where $a$ is the
      arity of $R$, $v_1, \dots , v_a \in V$, and $R \in \A$.}
{Is there a function $f : V \rightarrow D$
          such that $(f(v_1), \dots , f(v_a)) \in R$ for every
          $R(v_1, \dots , v_a) \in \constraints$?} 
          
The function $f$ is said to be a a {\em satisfying assignment} or simply a {\em solution}.   
Having a multiset of constraints is a simple way of introducing a weighting
mechanism for constraints in the forthcoming \maxcsp\ problem; sometimes more elaborate ways are considered but this
one is sufficient for our needs.

The {\em value} of an assignment $\varphi$ for $I=(V,C)$ is the number of constraints in $C$ satisfied by $\varphi$.
The \emph{maximum constraint satisfaction problem} ($\maxcsp(\A)$)
is defined as follows:

\pbDef{\maxcsp(\A)}
{An instance $(V,C)$ of $\csp(\A)$ and an integer $m$.}
{Is there a set $X \subseteq C$ such that $\abs{X} \geq m$ and $(V,X)$ is satisfiable?}

Given an instance $((V,C),m)$ of $\maxcsp(\A)$, the set 
$X$ can be computed 
with $|C|$ calls to an algorithm for $\maxcsp(\A)$. 
Hence, we can view $\maxcsp(\A)$ as a decision problem without loss of generality.

The {\em primal graph} of a CSP instance $(V,C$) has vertex set $V$ and
and two vertices are joined if their corresponding variables occur together 
in a constraint in $C$. 
We view the primal graph as a simple undirected graph, even though it is
technically a graph with loops on every vertex 
(equivalently, one may only join distinct vertices).

\subsection{Relations and Constraint Languages}
\label{sec:relations}

Dechter et al.~\cite{Dechter:etal:ai91} introduced the highly influential {\em simple
temporal problem}.
Let $\textbf{S}_b$ be the set of binary relations
\[ B_{a,b} = \{ (x_1,x_2) \in \rationals^2 \mid a \leq x_1 - x_2 \leq b \} \]
for endpoints $a \in \integers \cup \{-\infty\}$ and
$b \in \integers \cup \{\infty\}$ such that $a \leq b$,
$(a,b)\neq(-\infty,\infty)$.
Let $\textbf{S}_u$ be the set of unary relations
$U_{a,b} = \{ (x) \in \rationals \mid a \leq x \leq b \}$
for endpoints $a \in \integers \cup \{-\infty\}$ and
$b \in \integers \cup \{\infty\}$ such that $a \leq b$,
$(a,b)\neq(-\infty,\infty)$.
$\csp(\textbf{S})$ where $\textbf{S}=\textbf{S}_b \cup \textbf{S}_u$ is known as the {\em simple temporal problem} ($\probstp$).
To aid readability, we sometimes denote 
the constraint $R_{a,b}(x,y)$ by $a \leq x - y \leq b$.
We follow e.g.~\cite{Tsamardinos:Pollack:ai2003} and assume
that bounding values are integers.
This implies that every
satisfiable instance of
$\probstp$ admits an integer solution (see Dechter et al.~\cite[Section 3]{Dechter:etal:ai91}).
We can thus work over $\integers$ without loss of generality.

Given an \probstp\ relation $R_{a,b}$, we let $\magni{R_{a,b}}=\max(\{|a|,|b|\} \setminus \{\infty\})$. 
We say that $\magni{R_{a,b}}$ is the {\em magnitude} of $R_{a,b}$.
Note that the magnitude of, for instance, $R_{1,\infty}$ is 1.
If $X$ is a set of relations or constraints, then the definition of $\magni{\cdot}$ extends naturally: $\magni{X} = \max_{R \in X} \magni{R}$, i.e. $\magni{X}$ is the
least
upper bound on absolute values of all numerical bounds appearing in the relations of $X$.

Let $\maxstp$ denote the problem $\maxcsp(\textbf{S})$ and $\maxstp_b$ denote
$\maxstp$ restricted to binary relations.
We will look at restricted versions of $\maxstp$
so we let $\textbf{S}^{(k)}$, $k \in \naturals$, denote $\textbf{S}$ restricted to relations $R$
satisfying $\magni{R} \leq k$ and let
$\maxstp^{(k)}$
denote $\maxcsp(\textbf{S}^{(k)})$
(and $\maxstp_b^{(k)}$ is defined in the obvious way).

We finally note that the CSP$(\textbf{S}_b)$ problem
is invariant under translation, i.e. 
if $\varphi : V \rightarrow \integers$ satisfies an instance $(\variables, \constraints)$, 
then so does $\varphi'(v) = \varphi(v) + c$ 
for all $v \in \variables$ and an arbitrary $c \in \integers$.
Thus, we can pick any variable in $\variables$ 
and assume that its value is zero without loss of generality.
We call such a variable a \textit{zero variable}.
Augmenting an instance of \probstp\ with a zero variable $z$
allows us to express unary constraints in $\textbf{S}_b$, e.g. the constraint $0 \leq x - z \leq 2$ 
is equivalent to $x \in \{0,1,2\}$.

\subsection{Parameterized Complexity}
\label{sec:para-complexity}

We use the framework of 
{\em parameterized complexity}~\cite{DowneyFellows13,book/FlumG06,book/Niedermeier06},
where the run-time of an algorithm is studied with respect to a parameter
$p\in\Nat$ and the input size~$n$. A parameterized problem is then a subset of $\Sigma^* \times {\mathbb N}$
(where $\Sigma$ is the input alphabet).
The idea is that the parameter describes the structure of
the instance in a computationally meaningful way.
Here, the most favourable complexity class is \FPT,
which contains all problems that are \emph{fixed-parameter
  tractable (\FPT)}, i.e. can be decided 
in $f(p)\cdot n^{O(1)}$ time,  where $f$ is a computable
function.
The next best option is the
complexity class \XP, which contains all problems decidable
in $n^{f(p)}$ time, i.e. the problems solvable in polynomial time
when the parameter $p$ is bounded.
Clearly, $\FPT \subseteq \XP$ and
this inclusion is strict 
(see e.g.~\cite[Cor. 2.26]{book/FlumG06}).
It is significantly 
better if a problem is in \FPT than in \XP
since the order of the polynomial factor in the former case does not depend on the parameter $p$.
Finally, the class \paraNP\ contains all problems that can be decided 
in $f(p)\cdot n^{O(1)}$ time by a non-deterministic algorithm for some computable function $f$.
It is known that
a problem is \paraNP{}\hy hard (under \FPT{\em -reductions}; see
below) if it is \NP{}\hy hard for some constant value of the parameter. 
Problems that are \paraNP{}\hy hard are considered to be
significantly harder than those in \XP since a problem that is
\paraNP{}\hy hard cannot be in \XP unless \PP = \NP.

Reductions between parameterized problems must take
the parameter into account. To this end, we use {\em parameterized reductions} (or \FPT-reductions).
Let $L_1$ and $L_2$ denote parameterized problems with $L_1 \subseteq \Sigma_1^* \times {\mathbb N}$
and $L_2 \subseteq \Sigma_2^* \times {\mathbb N}$. 
A parameterized reduction from $L_1$ to $L_2$ is a
mapping $P: \Sigma_1^* \times {\mathbb N} \rightarrow \Sigma_2^* \times {\mathbb N}$
such that

\smallskip

\noindent
1.
  $(x, k) \in  L_1$ if and only if $P((x, k)) \in L_2$, 

\smallskip

 \noindent
 2.
  the mapping can be computed by an \FPT-algorithm 
  with respect to the parameter $k$, and 

\smallskip

  \noindent
  3.
  there is a computable function $g : {\mathbb N} \rightarrow {\mathbb N}$ 
such that for all $(x,k) \in L_1$ if $(x', k') = P((x, k))$, then $k' \leq g(k)$.
\smallskip

\noindent
The class $\Weft[1]$ contains all problems that are \FPT-reducible to \textsc{Independent Set} parameterized by
the size of the solution set.
Showing $\Weft[1]$-hardness (by an \FPT-reduction) for a problem rules out the existence of an \FPT\
algorithm under the assumption $\FPT \neq \Weft[1]$.
To obtain sharper bounds we sometimes use the
{\em exponential-time hypothesis} (ETH), which is
an ubiquitous
computational hardness assumption, implying that satisfiability of 3-CNF Boolean formulas (3-SAT) with $n$ variables cannot be solved in $2^{o(n)}$ time.

\paragraph{Tree Decompositions and Treewidth}
A central parameter is \emph{treewidth} which intuitively 
describes how far a graph is from being a tree.
While the name was introduced by~\cite{Robertson:Seymour:jctb84}, the idea appeared much earlier in nonserial dynamic programming~\cite{Bertele:Brioschi:NDP72} and have been intensively applied in AI (e.g. bucket elimination~\cite{Dechter99}).
A tree decomposition $\TTT=(T, \chi)$ of a graph $G = (V, E)$ 
consists of a rooted tree $T$ and a mapping $\chi$ that assigns each node $t \in V(T)$ a set $\chi(t)\subseteq V$, called \emph{bag}, such that 
$E(G) \subseteq \{ \{u,v\} \,|\, t \in V(T), \{u,v\} \subseteq \chi(t)\}$; and 
(ii)~for every $v \in V$, the nodes, for which the bag contains~$v$, 
form a non-empty sub-tree of $T$.
We let $\text{width}(\TTT) = \max \{
\abs{\chi(t)}-1 : t \in T \}$. 
The {\em treewidth} of a graph $G$,
denoted by $\tw(G)$, is the minimum $\text{width}(T)$ over all tree decomposition of~$G$. 
Given an instance $I$ of CSP, we use the phrase {\em primal treewidth}
for the treewidth of the primal graph of $I$.
For arbitrary but fixed $w \geq 1$, we can decide in linear time whether a graph has treewidth at most $w$ and, if so, to compute a tree decomposition of width $w$~\cite{Bodlaender:sicomp96}.

\section{Bounded Magnitude is Necessary} \label{sec:unbounded_magnitude}

Exhaustive enumeration of all subsets of constraints solves $\maxstp$ in $O^*(2^{|C|}) \subseteq O^*(2^{||I||})$ time where $||I||$ is the number of bits required to represent the instance $I$. Since $||I||$ can be \emph{much} larger than the number of variables $n$ this bound says virtually nothing. We begin this section by constructing an algorithm with running time dominated by $\magni{C}^n$, i.e., $2^{n \log \magni{C}}$.
Thus, $\maxstp$ is solvable in single-exponential time for any fixed magnitude. 

\begin{theorem} \label{thm:time-upperbound}
    $\maxstp_b$ can be solved in $O^*((k+3)^n)$ time, where $n=|V|$ and $k=\magni{C}$.
\end{theorem}
\begin{proof}
    We use dynamic programming and construct all relative positions of variables with respect to some boundary value moving to the right on the natural line.
    Specifically, we keep track of a partition of the $n$ vertices in $k+3$ groups: one group $A_<$ containing vertices more than $k$ below the boundary value, $k+1$ groups $A_0, \ldots, A_k$ containing variables exactly $0, \ldots, k$ below the boundary value, and one group $A_>$ containing vertices at or above the boundary value. Note that the vertices at the boundary value can be either in $A_0$ or in $A_>$. Overall, this results in $(k+3)^n$ states.
    We call variables in $A_<,A_0,\ldots,A_k$ as \emph{decided} and variables in $A_>$ as \emph{undecided}. 
    
    The algorithm starts in the state where $A_< = A_0=\ldots = A_k=\emptyset$ and $A_>=V$ and recursively moves variables from $A_>$ to $A_<$, maintaining that for each configuration $(A_<,A_0,\ldots, A_k,A_>)$, the maximum number of satisfied constraints between decided variables is saved. 
    In each step, the two following transitions are considered:
    
    \smallskip
    
    \noindent
    {\em Transition 1}. We pick any undecided variable $v$ and assign it the current boundary value. This moves it from group $A_>$ to group $A_0$. Since $v$ is now decided, the subproblem needs to account for constraints involving $v$ and other decided variables $u$. 
    For all $u \in A_i, 0 \leq i \leq k$, we have $v-u = i$. For all $u \in A_<$, we have $v-u > k = \magni{C}$. In both cases we know whether a constraint between $v$ and $u$ is satisfied.
    
    \smallskip
    
    \noindent
    {\em Transition 2}.
    We increase the boundary by one. This merges $A_<$ and $A_k$ into a new group $A_<$, relabels each group $A_i$ as $A_{i+1}$ for $i \in \{0, \ldots, k-1\}$, and adds a new empty group $A_0$. Since this transition does not change the decided variables or their relative differences, the solution to the corresponding subproblem is unchanged.
    We disable this transition when the groups $A_0, \ldots, A_k$ are all empty as we can assume that without loss of generality, the gap between two variables will be at most $k+1$ by $k = \magni{C}$.
    
    \smallskip
    
    Clearly, each path from the starting state to the final state ($A_< = V$) using these transitions corresponds to a potential solution to the $\maxstp_b$ instance, and each potential solution corresponds to such a path. We now show in which order to visit all the states. Consider the quantity
    \[\abs{A_0} + 2\abs{A_1} + \ldots + (k+1)\abs{A_k} + (k+2)\abs{A_<}.\]
    In the first transition, we add one variable to $\abs{A_0}$, hence this quantity increases by one. In the second transition, each variable from the groups $A_0, \ldots, A_k$ moves to a group with a larger effect on this quantity, hence, outside of the edge case where all these groups are empty, it always increases. This shows that, when we visit the states ordered by this quantity, breaking ties arbitrarily, we can only visit a state after all states with a transition towards it have already been considered.

    Overall, this shows that the DP solves the $\maxstp_b$ instance. Since there are $(k+3)^n$ states, each state has at most $n+1$ transitions, and each transition can be computed in linear time, the total runtime is $O^*((k+3)^n)$.
\end{proof}

We complement this upper bound by a matching lower bound (under standard complexity theoretical assumptions) for the parameter $n = |V|$. Here, it is important to recall that $n$ is a very strong parameter that gives trivial \FPT\ algorithms for all finite-domain CSPs and infinite-domain qualitative reasoning problems. As such, it is a very useful lower bound since it can also (as we will show in Theorem~\ref{thm:no_graph_parameters}) be used to obtain lower bounds for \emph{all} reasonable graph parameters.
The heart of our reduction is the \MCC problem.

\pbDefP{\MCC}
{A graph $G = (V,E)$ with a vertex coloring $f_C \colon V \to C$ to some color set $C$ such that each edge is between vertices of different colors.}
{The number of colors $\abs{C}$}
{Is there a clique in $G$ with exactly one vertex from each color?} 

\begin{theorem}[Theorem~5.2 in \cite{Lokshtanov:etal:beatcs2011}]
    \label{thm:mcclique-hard}
    \MCC is (1) \W{1}-hard and (2) cannot be solved in $f(c)n^{o(c)}$ time unless the ETH fails, where $c = \abs{C}$ is the number of colors, $n = \abs{V}$ the number of vertices, and $f : \naturals \to \naturals$ any computable function.
\end{theorem}

Our reduction is based on {\em Sidon sets}: a set $S$
of integers such that the sum of any pair of its elements is unique, i.e.  if $a + b = c + d$ for
$a, b, c, d \in S$, then $\{a, b\} = \{c, d\}$.
Differences are easier to use in our proofs so we 
use an equivalent condition: 
for all $a,b,c,d \in S$ such that $a \neq b$ and $c \neq d$, $a-b=c-d$ holds
if and only if $a=c$ and $b=d$.
The {\em order} of a Sidon set is the number of elements in it and the {\em length} is the difference between its maximal and minimal elements.
For example, $\{0,1,4,6\}$ is a Sidon set (see Figure~\ref{tb:sidon}) of order $4$ with length $6$.
It is well-known that "reasonable" Sidon sets are
computable in polynomial time. For instance,
Dabrowski et al.~\cite[Section~2.3]{Dabrowski:etal:difflogicreport} describe
how a Sidon set of order $k$ and length $8k^2$ can be computed in polynomial time.

\begin{table}

\begin{center}
\begin{tabular}{|c|c|c|c|c|}\hline
      & $b=0$ & $b=1$ & $b=4$ & $b=6$\\ \hline
$a=0$ & 0     & -1    & -4    & -6\\
$a=1$ & 1     & 0     & -3    & -5\\
$a=4$ & 4     & 3     & 0     & -2\\
$a=6$ & 6     & 5     & 2     & 0 \\ \hline
\end{tabular}
\end{center}
\caption{The set $S=\{0,1,4,6\}$ is a Sidon set:
$a-b$ has distinct values whenever $a \neq b$ and $a,b \in S$.}
\label{tb:sidon}
\end{table}

\begin{theorem} \label{thm:generallowerbound}
$\maxstp_b$ (and thus $\maxstp$) is \W{1}-hard when parameterized by the number of variables. It cannot be solved in $2^{o(n\log k + f(n))}$ time unless the ETH fails where $n = \abs{V}$ is the number of variables, $k = \magni{C}$ the magnitude, and $f \colon \naturals \to \naturals$ any computable function.
\end{theorem}

\begin{proof}
We present a reduction from \MCC to an instance of $\maxstp_b$.
Let $(G=(V,E),f_C)$ be an instance of \MCC with $n$ vertices, $m$ edges and $c = |C|$ colors. We begin by introducing variables $x_1,\dots,x_c$ for each color and a zero variable $z$. We construct a Sidon set $S=\{s_1,\dots,s_n\}$ and add the constraints 
$x_{f_C(1)}-z=s_1,  x_{f_C(2)}-z = s_2,  \dots  x_{f_C(n)}-z=s_n,$
which we make costly to delete by repeating them $m+1$ times. Additionally, for each edge $(u,v) \in E$, we add the constraint $x_{f_C(u)}-x_{f_C(v)}=s_{u}-s_{v}$. This completes our reduction. Note that we now have $(m+1)n + m$ constraints.

Let $W := \{z,x_1,\dots,x_c\}$ the resulting set of variables and $D$ the resulting set of constraints. Clearly, $(W,D)$ can be computed in polynomial time. We now verify that $G$ contains a multicolored clique if and only if $((W,D),N)$ is a yes-instance of $\maxstp_b$ where $N := c(m+1) + \binom{c}{2}$.

\smallskip

\noindent
{\em Forward direction.}
Assume $G$ contains a multicolored clique $K$. Let $g \colon C \to K$ denote for each color the vertex from $K$ that has this color. 
Define $h \colon W \to \integers$ such that $h(z)=0$ and $h(x_i)=s_{g(i)}$. We show that $h$ is a solution to the $\maxstp_b$ instance. Each edge $(u,v) \in K$ corresponds to the constraint $x_{f_C(u)}-x_{f_C(v)}=s_{u}-s_{v}$. Since $h(x_{f_C(u)})-h(x_{f_C(v)}) = s_{g(f_C(u))} - s_{g(f_C(v))} = s_u - s_v$, these $\binom{c}{2}$ constraints are satisfied. For each vertex $v \in K$, consider the constraint $x_{f_C(v)}-z = s_v$, which was repeated $m+1$ times. Since $h(x_{f_C(v)}) - h(z) = s_{g(f_C(v))}-0 = s_v$, we also satisfy this, resulting in an additional $c(m+1)$ satisfied constraints. In total, we satisfy at least $N$ constraints.

\smallskip

\noindent
{\em Backward direction.}
Assume $((W,D),N)$ is a yes-instance and let $h \colon W \to \integers$ be a solution. By globally adding a suitable constant, we assume that $h(z)=0$.

For each variable $x_i$, all constraints between $x_i$ and $z$ are disjoint up to the $m+1$ time repetition, hence we satisfy either 0 or $m+1$ such constraints. Since there are only $m$ constraints not of this form, we must satisfy $m+1$ such constraints for each variable; if not, we satisfy at most $(c-1)(m+1) + m < c(m+1) < N$ constraints.
It follows that $x_i$ is assigned a value $s_u$ for some $u \in V$ with $f_C(u) = i$. Let $g \colon C \to V$ be the function mapping each $i$ to the corresponding $u$.

Now consider two variables $x_i$ and $x_j$. As $S$ is a Sidon set, and both $x_i$ and $x_j$ are assigned values from $S$, every constraint of the form $x_{f_C(u)}-x_{f_C(v)}=s_{u}-s_{v}$ with $f_C(u) = i$ and $f_C(v) = j$ is only satisfied if $x_i$ is assigned $s_u$ and $x_j$ is assigned $s_v$. It follows that we can only satisfy at most one such constraint between $x_i$ and $x_j$.

In order to satisfy $c(m+1) + \binom{c}{2}$ constraints in total, we must satisfy exactly one such constraint for each pair $(x_i,x_j)$ as there are no other remaining constraints. This shows that, for each $x_i$ and $x_j$, the constraint $x_i - x_j = s_{g(i)} - s_{g(j)}$ must exist. That is, $g(i)$ and $g(j)$ are connected in the original graph. Hence, the image of $g$ is a multicolored clique.

\smallskip

\noindent
Since the total number of variables in $W$ is $n' := c+1$, the reduction is \FPT. This shows \W{1}-hardness. Since the absolute values of the integers appearing in the constraints do not exceed $k' := s_{n} < 8n^2$, any algorithm solving $\maxstp$ in $2^{o(n'\log k' + f'(n'))}$ time for some $f'$ can be used to solve \MCC in
\[
    2^{o((c+1)\log(8n^2) + f'(c+1))} = 2^{f'(c+1)} \cdot n^{o(c)}
\]
    time. This contradicts the ETH by Theorem~\ref{thm:mcclique-hard}.
\end{proof}

This lower bound has significant consequences for virtually all other parameters.

\begin{theorem} \label{thm:no_graph_parameters}
If $\kappa$ is a computable graph parameter, then $\maxstp$ is \W{1}-hard when parameterized by $\kappa$ of primal graph.
\end{theorem}
\begin{proof}
As $\kappa$ is computable, the function $f \colon \naturals \to \naturals$ where $f(k)$ is the $\kappa$ of a clique on $k$ vertices is computable.
Let $G=(V,E)$ be an instance of the \MCC problem with $n'$ variables and $k'$ colors.
We use the same reduction as in the proof of Theorem~\ref{thm:generallowerbound}.
The primal graph of $(W,D)$ is now a clique on $n = k'+1$ vertices hence
$\kappa$ of primal graph is $f(k'+1)$, so the reduction is \FPT.
As \MCC is \W{1}-hard, this gives \W{1}-hardness for $\maxstp$.
\end{proof}

\section{Parameterized Complexity Results}
\label{sec:boundedmagnitude}

Section~\ref{sec:unbounded_magnitude} shows that any reasonable parameterization needs to take the magnitude into account. Thus, we now consider the parameterized complexity with parameter magnitude. 
A straightforward reduction shows the following \paraNP-hardness result.

\begin{theorem} \label{thm:mag-pnp}
   $\maxstp_b$ is \paraNP{}\hy hard when parameterized by ${\sf mag}$.
\end{theorem}

\begin{proof}
We use
the following \NP-hard problem~\cite{Ka72}.

\pbDef
{\textsc{Maximum Acyclic Subgraph} (MAS)}
{A directed graph $D = (V,A)$ and an integer $m$.}
{Is there a set $X \subseteq A$
such that $|X| \geq m$ and $(V, X)$ is acyclic?}

Let $((V,A),m)$ be an arbitrary instance of MAS.
Construct an instance $((V,C),m)$ of $\maxstp_b$ where each arc $vw \in A$
is replaced by the constraint $v - w \leq -1$.
This constraint can be viewed as enforcing that $f(v) < f(w)$
for any solution $f$.
It is clear
that $((V,A),m)$ is a yes-instance if and only if $((V,C),m)$ is
a yes-instance. Furthermore, $\magni{C} \leq 1$ so $\maxstp_b$ is 
indeed \paraNP{}\hy hard when parameterized by magnitude. 
\end{proof}

Hence, the magnitude does not help in itself, but it can be combined with other parameters. We explore this in Section~\ref{sec:vertexcover} where we obtain a positive \FPT\ result based on vertex cover, and in Section~\ref{sec:treewidth-algo} where we consider treewidth. The latter parameterization turns out to be \W{1}-hard but can be used for an \XP algorithm.

\subsection{Vertex Cover}
\label{sec:vertexcover}

A {\em vertex cover} of a (directed or undirected) graph is a set of vertices that includes at least one endpoint of every edge. Given an instance $I=(V,C)$
of $\csp$, we let $\vc(I)$ denote the size of the smallest vertex
cover of the primal graph of $I$.
We exhibit a connection between the structure of \probstp\ instances
and the size of solution domains.
Consider the directed edge-weighted \emph{distance graph} $\Delta_I = (V, A, w)$.
This graph is constructed by for each constraint $a \leq x - y \leq b$ in $C$,
adding arcs $(x,y)$, $(y,x)$ to $A$ with weights $w(x,y) = b$, $w(y,x) = -a$.
Arcs of infinite weight are not added to $A$. 
Dechter et al.~\cite[Theorem~3.1]{Dechter:etal:ai91})
have noted the following: an
  instance $I$ of $\csp{(\S)}$ is satisfiable if and only if 
  $\Delta_I$ contains no directed cycle of total negative weight.

\begin{lemma}\label{lem:stp-smallsolution-diameter}
    Let $I = (V,C)$ be an instance of $\stp_b$ where $k = \magni{C}$ is its magnitude and $d$ is the length of
the longest simple path
    in the distance graph $\Delta_I$.
     If $I$ has a solution, it has one where each value comes from the set
    $\left\{0, 1, 2,\ldots, dk\right\}$.
\end{lemma}
\begin{proof}
    Consider the distance graph $\Delta_I$. Since $I$ has at least one solution, $\Delta_I$ has no directed cycles of negative weight. Modify $\Delta_I$ by adding a zero vertex $z$ together with an arc of weight zero from each other vertex to $z$ and let $f \colon V \to \integers$ be the function where, for each $x \in V$, $f(x)$ is the length of the shortest path from $x$ to $z$, measured by arc weight. This shortest distance exists since $\Delta_I$ has no directed cycles of negative weight. Note that $f$ is always non-positive since each shortest path is not longer than the path obtained by directly following the arc towards $z$. We now show that $-f$ is a solution satisfying the lemma.
    
    First, $-f$ is a solution to the instance: if some constraint $x-y \leq c$ is not satisfied, then the path from $x$ to $z$ could be made shorter by going via $y$. Furthermore, $-f$ is at least zero since $f$ is non-positive. Finally, $-f$ is at most $dk$ since each shortest path consists of at most $d$ arcs and each arc decreases the total length by at most $k$. This completes the proof.
\end{proof}

A consequence of Lemma~\ref{lem:stp-smallsolution-diameter}
(which we will use in Section~\ref{sec:treewidth-algo})
is that solvable \probstp\ instances $(V,C)$
have a solution that only uses values from the set 
$\left\{0, 1, 2,\ldots, (n-1)k\right\}$ where $n=|V|$ and $k=\magni{C}$. 

\begin{theorem}
    $\maxstp_b$ can be solved in $\bigoh^*((2\vc \cdot k)^\vc n)$, where $n$ is the number of vertices and $k$ the magnitude. In particular, it is \FPT\ when parameterized by ${\sf mag} + \vc$.
\end{theorem}
\begin{proof}
    A graph with vertex cover $\vc$ has a longest simple path of length at most $2\vc+1$: if there exists a longer path, then this path must contain two adjacent vertices that are not part of the vertex cover, and thus it does not cover all edges. Since the primal graph and the distance graph have the same vertex cover, Lemma~\ref{lem:stp-smallsolution-diameter} implies that 
    a solvable instance has a solution with values from $A = \{0,  \ldots, (2\vc+1)k\}$.
    
    We now describe our algorithm. Let $W$ with $\abs{W} = \vc$ be a vertex cover. For each variable in $W$,  guess a value from $A$. This results in ${\abs{A}}^\vc$ branches. In each branch, we count how many constraints between vertices from $W$ are satisfied.
    We then determine for every other variable $v$ the largest number of constraints between $v$ and $W$ that can be satisfied by trying all $\abs{A}$ options for $v$. Since there are no edges between vertices outside of $W$, these choices are all independent, hence we find the optimal value in this branch. By computing the maximum number of satisfied constraints over all branches, we find a solution to the instance.
    The overall runtime is
    \[\bigoh\left(\abs{A}^{\vc+1}n\right) = \bigoh\left(((2\vc+1)k)^{\vc+1} n\right) = \bigoh^*\left((2\vc k)^\vc n\right).\]
\end{proof}

\subsection{Treewidth}
\label{sec:treewidth-algo}

We show that $\maxstp_b^{(k)}$ is \W{1}-hard when parameterized by treewidth of primal graph for every $k \geq 1$ ($\maxstp_b^{(0)}$ is trivial since every constraint can be satisfied by assigning zero to all variables).
Nevertheless, treewidth can be used to solve $\maxstp$ faster, and we construct an \XP algorithm running in roughly $(nk)^{\tw}$ time where $k = \magni{C}$ and $n = |V|$.

\subsubsection{\W{1}-hardness}

Our \W{1}-hardness result is based on the proof of Theorem \ref{thm:generallowerbound}.

\begin{theorem} \label{thm:stp-w1hard}
$\maxstp_b^{(1)}$ is \W{1}-hard when parameterized by $\tw$.
\end{theorem}
\begin{proof}
    Recall the reduction from Theorem \ref{thm:generallowerbound}. It results in a \W{1}-hard instance of $\maxstp_b$ whose primal graph is a clique on $c+1$ vertices and each edge corresponds to an equality constraint of magnitude at most $8n^2$, where $c$ is the parameter and $n$ is bounded by instance size. There can be multiple edges between the same vertices, but there are $\bigoh(mn)$ edges in total, where $m$ is again bounded by instance size.

    We transform this instance into one with magnitude 1 by subdividing edges. Consider an edge $(x,y)$ corresponding to a constraint $f(x)-f(y) = a$ with $\abs{a} \geq 2$. Assume without loss of generality that $a$ is positive. We replace this edge with a path of $a$ edges on vertices $x_0,x_1,\ldots,x_a$ where $x_0=x$, $x_a=y$, and $x_1,\ldots,x_{a-1}$ are fresh variables. Each edge $(x_{i-1},x_i)$ is given the constraint $f(x_{i-1})-f(x_i) = 1$. The result is a graph on $\bigoh(mn\cdot n^2)$ vertices with magnitude 1.

    This new instance is equivalent to the original instance: an optimal solution to the new instance will break at most one edge in each path, and breaking one edge in a path is equivalent to breaking the original high-magnitude edge.

    Finally, we show that the resulting graph has bounded treewidth by giving a tree decomposition of size $\max(2,c)$, which implies that the reduction is \FPT\ and hence that $\maxstp_b^{(1)}$ is \W{1}-hard when parameterized by $\tw$. We begin with one bag $B$ containing all $c+1$ vertices from the original clique. For each path $x_0,x_1,\ldots,x_a$, we add a path of bags $B_1,\ldots, B_a$ with $B_i := \{x_{i-1},x_i,y\}$ for $1 \leq i \leq a$ and where $B_1$ is adjacent to $B$. This results in a valid tree decomposition. The size of the largest bag is $\max(3,c+1)$, hence the width is $\max(2,c)$, which completes the proof.
\end{proof}

\subsubsection{\XP Algorithm}

We continue with our \XP algorithm. 

\begin{theorem}\label{thm:algorithm-longest-path}
    $\maxstp_b$ can be solved in $\bigoh^*((nk)^{\tw})$, where $n$ is the number of vertices, $k$ the magnitude, and $\tw$ the treewidth of the distance graph $\Delta_I$.
\end{theorem}
\begin{proof}
    We begin with a definition.
        Let $T$ be a tree decomposition. We say that $T$ is \emph{nice} if it is rooted using some root $r$ with $\chi(r) = \emptyset$, and each node $t$ has one of the following types:
        \begin{itemize}
            \item \emph{leaf node}: $\chi(t) = \emptyset$ and $t$ has no children.
            \item \emph{introduce node}: $t$ has one child $t'$, and there exists a vertex $v \in \chi(t)$ such that $\chi(t') = \chi(t) \setminus \{v\}$.
            \item \emph{forget node}: $t$ has one child $t'$, and there exists a vertex $v \in \chi(t')$ such that $\chi(t) = \chi(t') \setminus \{v\}$.
            \item \emph{join node}: $t$ has two children $t_1,t_2$ and $\chi(t) = \chi(t_1) =bi \chi(t_2)$.
        \end{itemize}
    Any tree decomposition can be modified into a nice tree decomposition of the same width in linear time. For a node $t \in T$, we define $T_t$ as the subtree rooted at $t$, and define $\chi(T_t) = \bigcup_{t' \in T_t}\chi(t')$ as the set of all variables occurring in this subtree~\cite{BodlaenderKoster08}.

    We use standard dynamic programming over a tree decomposition ~\cite{SamerSzeider10}. Let $T$ be a nice tree decomposition of the primal graph $\Delta_I$ with width $\tw$. By a defining property of tree decompositions, for each constraint $C$, there exists at least one node whose bag contains all variables occurring in $C$. However, there may be several such nodes. To avoid double counting, we assign $C$ to one such node, chosen arbitrarily. For each node $t$, we let $C_t$ denote the constraints associated with $t$.

    By Lemma \ref{lem:stp-smallsolution-diameter}, noting that any simple path has length at most $n-1$, we only need to look for solutions from the set $S = \left\{0, 1, 2,\ldots, (n-1)k\right\}$. We now wish to compute the following information: for each node $t$, and each possible assignment $f\colon \chi(t) \to S$, what is the largest number of constraints that can be satisfied in any extension of $f$ into an assignment $f'\colon \chi(T_t) \to S$ where we only consider constraints that are associated with nodes from $T_t$? Let $DP(t,f)$ denote this value.

    Given $t$ and $f$, we note that $DP(t,f)$ equals the number of constraints in $C_t$ satisfied by $f$ plus the following depending on the type of $t$.
    \begin{itemize}
        \item \emph{leaf node}: $T_t$ contains no other constraints, hence we add nothing.
        \item \emph{introduce node}: let $t'$ be the child node. Any extension of $f$ must extend the restriction $\left.f\right|_{\chi(t')}$, hence we add the value $DP(t',\left.f\right|_{\chi(t')})$.
        \item \emph{forget node}: let $t'$ be the child node and $v$ the forgotten variable. For each $i \in S$, let $f_i$ be the extensions of $f$ into an assignment of $\chi(t')$ where $f_i(v) = i$. Since any extension of $f$ in $\chi(T_t)$ extends some $f_i$, we add the maximum $\max_{i \in S}(DP(t',f_i))$.
        \item \emph{join node}: let $t_1$ and $t_2$ be the two child nodes. Since any extension of $f$ to $\chi(T_t)$ can be split into two independent extensions $f_1$ in $\chi(T_{t_1})$ and $f_2$ in $\chi(T_{t_2})$, we add the sum $DP(t_1,f) + DP(t_2,f)$.
    \end{itemize}
    Finally, since the root $r$ is empty, $DP(r,\emptyset)$ holds the answer to the entire instance. For all nodes $t$ except forget nodes, there are at most $\abs{S}^{\tw+1}$ states and each state can be computed in $O(\abs{C_t})$ time. For forget nodes $t$, each state requires $O(\abs{S} + \abs{C_t})$ time, but we have $\abs{\chi(t)} < \tw+1$ since the child node is larger and has size at most $\tw+1$, hence computing all states also requires $\bigoh(\abs{S}^{\tw+1} \abs{C_t})$ time. Since there are $\bigoh(n)$ nodes total, and all $C_t$ sum to $m$, the total runtime is $\bigoh(\abs{S}^{\tw+1}(n+m)) = \bigoh((nk)^\tw)$.
\end{proof}

\section{A Comparison with Qualitative Reasoning} \label{sec:qualitative}

Qualitative reasoning is an influential subarea of AI where quantitative (e.g., numerical) relations are
avoided in favour of qualitative relations. 
Well-known qualitative formalisms include Allen's Interval algebra~\cite{allen1983maintaining},
the spatial RCC formalisms~\cite{Randell:etal:kr92}, and various cardinal direction calculi~\cite{Frank:ogai91,Goyal:PhD,Ligozat:vlc98}.
The intersection between qualitative reasoning and CSPs has been intensively studied, generating a large number 
of formalisms and results~\cite{Bodirsky:Jonsson:jair2017,Dylla:etal:acmsurv2017}.
Clearly, \probstp\ is 
is {\em not} qualitative since  it is profoundly based in relations between numerical values.
It is thus interesting to 
make comparisons between the \maxcsp\ for qualitative
formalisms and \maxstp. 

One approach~\cite{ijcai2023p212,ErikssonLagerkvist2023} for solving spatio-temporal CSPs is to construct the respective spatio-temporal orders using branching and merging two nodes in the branching tree if for all expansions of the respective partial orders corresponding to the nodes to total orders, either both total orders correspond to satisfying assignments or neither.
We find that dynamic programming (DP) approaches exploiting this idea are often generalizable to \maxcsp.
Furthermore, Dabrowski et al.~\cite{Dabrowski:etal:ai2023} have proven that $\csp(\A)$ for constraint languages with the {\em patchwork} property (see~\cite{Lutz:Milicic:jar2007})
 are \FPT\ parameterized by the treewidth of the primal graph. Patchwork states that the union of two
satisfiable \csp\ instances whose constraints agree on their common
variables is satisfiable. 
This result generalizes to \maxcsp, and
shows, for instance, that \maxcsp(Allen) and \maxcsp(RCC8) are in \FPT\ when parameterized by $\tw$. We conclude the following.
\begin{obs*} 
Let $\A$ be a finite
constraint language with 
jointly-exhaustive and pairwise-disjoint relations such that CSP$(\A)$ is decidable. 

\smallskip

\noindent
\textbf{(1).}
    If $\csp(\A)$ is decidable in a time $f(n)$ for some $f(n) \in \Omega^*(\exp(n))$ using DP and a search tree such that 
    \begin{itemize}
        \item any two nodes in the tree that are equivalent to each other ($v \sim w$ are merged in the branching procedure) are also equivalent to each other with respect to the number of satisfied constraints, i.e. appending the same path to both nodes will result in the same number of constraints being satisfied by the solutions of the respective leaves,
        \item the number of constraints satisfied by each node is computable in polynomial time from parent nodes,
        \item for each assignment of variables, there is a corresponding node in the search tree (in particular also unsatisfying assignments),
    \end{itemize}
    then $\maxcsp(\A)$ is solvable in time $O^*(f(n))$.

    \smallskip

    \noindent
    \textbf{(2).} If $\A$ has the patchwork property and $\B$ is a finite structure whose relations are
Boolean combinations of relations in $\A$ (i.e. relations definable
by quantifier-free formulas that only contain the relations in~$\A$), then the problem
$\maxcsp (\B)$ is \FPT\ parameterized by the treewidth of the primal graph.
\end{obs*}

Consider Algorithm~2 from~\cite{ErikssonLagerkvist2023} that solves \csp(Allen) in time~$O^*((cn/\log n)^n)$ by generating all 
possible \textit{records} of a given instance. By altering the definition of a record slightly by including the constraints satisfied by the record in the record and altering the algorithm analogously to compute all possible records regardless of whether they contradict constraints, we construct an algorithm that computes \maxcsp(Allen) in time~$O^*((cn/\log n)^n)$.

Note, however, that in some instances we might need to reprove parts of the analysis that do not directly extend to \maxcsp. Consider e.g. Theorem~19 in~\cite{victorleifjohannaaai26} showing a $O^*((cn/\log n)^n)$ result for \csp(RCC-8). While the algorithm employs dynamic programming and is generalizable to \maxcsp, it reduces the problem to a tractable fragment of \csp(RCC-8) which is not proven to be tractable for \maxcsp. To extend this result, we would thus need to prove that the reduced fragment is also tractable for \maxcsp\ or at least solvable in $O^*((cn/\log n)^n)$.

\section{Concluding Remarks}
We studied the parameterized complexity of $\maxstp$. We began by giving a general \W{1}-hardness proof that extended to \emph{any} graph parameter. This necessitated a multi-parameterized with \emph{magnitude} being a promising starting point. Together with vertex cover size it can either be used directly for an \FPT\ algorithm, or, when combined with treewidth, for an \XP\ algorithm. 
Naturally, the map of the parameterized complexity landscape of \maxstp\ is not complete, and the complexity status of many interesting additional parametrizations is wide open. For example, can the \FPT\ algorithm for $\vc$ be extended to e.g.\ \emph{tree-depth}? If this is possible, \emph{path-width} would be a logical next step where it should be feasible to either extend the \FPT\ algorithm or to prove \W{1}-hardness (similar to treewidth). For \XP\ algorithms one could also investigate more general parameters such as \emph{clique-width} or \emph{twin-width}.

\section*{Acknowledgments}
Authors are given in alphabetical order.
The research of the first author was funded in whole or in part by
Excellence Center at Linköping -- Lund in Information Technology
(ELLIIT) funded by the Swedish government and the Wallenberg AI,
Autonomous Systems and Software Program (WASP) funded by the Knut and
Alice Wallenberg Foundation.
The third and fifth authors are partially supported  by the Swedish Research Council (VR) under grant 2021-04371. The fourth author is partially supported by VR under grant VR-2022-03214 and VR-2025-04487.

\bibliographystyle{abbrv}
\bibliography{references}

\end{document}